\documentclass[11pt]{article}
\usepackage{fullpage}
\usepackage{times}
\usepackage{microtype}
\usepackage{graphicx}
\usepackage{booktabs} 

\usepackage{algorithm}
\usepackage{algorithmic}
\usepackage{amsmath}
\usepackage{amssymb}
\usepackage{nicefrac}
\usepackage{amsthm}
\usepackage{hyperref}
\usepackage{enumerate}
\usepackage{appendix}

\usepackage{caption,wrapfig,bm}
\usepackage[usenames,dvipsnames,svgnames,table]{xcolor}
\usepackage{thm-restate}
\usepackage{listings}
\makeatletter
\newcommand{\setword}[2]{%
	\phantomsection
	#1\def\@currentlabel{\unexpanded{#1}}\label{#2}%
}
\makeatother

\DeclareMathOperator{\maxdet}{MAXDET}
\DeclareMathOperator{\maxvol}{MAXVOL}

\def\compactify{\itemsep=0pt \topsep=0pt \partopsep=0pt \parsep=0pt}
\let\latexusecounter=\usecounter

{\begin{itemize}%
\setlength{\itemsep}{0pt}%
\setlength{\topsep}{0pt}%
\setlength{\partopsep}{0pt}%
\setlength{\parskip}{0pt}}%
{\end{itemize}}

\hypersetup{%
  breaklinks,%
  ocgcolorlinks, %
  colorlinks=true,%
  urlcolor=[rgb]{0.45,0.0,0.0},%
  linkcolor=[rgb]{0.0,0.45,0.0},%
  citecolor=[rgb]{0,0,0.45}%
}%

\begin{document}

\title{Composable Core-sets for Determinant Maximization:  \\
            A Simple Near-Optimal Algorithm \footnote{This is an equal contribution paper. This paper has appeared in the 36th International Conference on Machine Learning (ICML), 2019}}


\author{ \begin{tabular}{ccccc}
{\begin{tabular}{c} Piotr Indyk \\ MIT \\\small{indyk@mit.edu} \end{tabular}} & & & &
{\begin{tabular}{c}Sepideh Mahabadi \\TTIC \\\small{mahabadi@ttic.edu} \end{tabular}}  \tabularnewline \\
{\begin{tabular}{c}Shayan Oveis Gharan \\University of Washington \\\small{shayan@cs.washington.edu } \end{tabular}} & & & &
 {\begin{tabular}{c}Alireza Rezaei \\University of Washington \\\small{arezaei@cs.washington.edu} \end{tabular}} 
\end{tabular}
}
\date{}

\newcommand{\tab}{\hspace{5mm}}

\newcommand{\eps}{\epsilon}
\newcommand{\sgn}{\operatorname{\text{{\rm sgn}}}}
\newcommand{\poly}{\operatorname{\text{{\rm poly}}}}
\newcommand{\floor}[1]{\lfloor #1 \rfloor}
\newcommand{\ceil}[1]{\lceil #1 \rceil}
\newcommand{\Fq}{\mathbb{F}_q}
\newcommand{\prob}[1]{P(#1)}
\newcommand{\Prob}[1]{P\left(#1\right)}
\newcommand{\Cspan}{\alpha}

\newtheorem{theorem}{Theorem}[section]
\newtheorem{remark}[theorem]{Remark}
\newtheorem{claim}[theorem]{Claim}
\newtheorem{lemma}[theorem]{Lemma}
\newtheorem{corollary}[theorem]{Corollary}
\newtheorem{conclusion}[theorem]{Conclusion}
\newtheorem{proposition}[theorem]{Proposition}
\newtheorem{conjecture}[theorem]{Conjecture}
\newtheorem{definition}[theorem]{Definition}
\newtheorem{observation}[theorem]{Observation}
\newtheorem{fact}[theorem]{Fact}
\newenvironment{proofof}[1]{{\vspace{3mm}\em Proof of #1.}}{\hfill
	\qed}

\def\LP{LP-based}
\def\U{{\cal U}}
\def\Re{{\mathbb{R}}}
\def\sphere{{\mathbb{S}}}
\def\Pr{{\mathbb{P}}}
\def\pntset{{P}}
\def\cs{{\mathcal{C}}}
\def\h{{\mathcal{H}}}
\def\ha{{\mathcal{G}}}
\def\sI{{\mathcal{I}}}
\def\Set#1{\{#1\}}
\def\abs#1{\left|#1\right|}
\def\norm#1{||#1||}
\def\dist{\mathrm{dist}}
\def\explain#1{\mbox{#1}}
\def\sS{{\mathcal{S}}}
\def\sC{{\mathcal{C}}}
\def\vol{{\mathrm{VOL}}}
\def\Vol{{\mathrm{VOL}}}
\def\alg{{\mathrm{ALG}}}

\maketitle

\begin{abstract}
``Composable core-sets'' are  an efficient framework for solving optimization problems in massive data models.
In this work, we consider efficient construction of composable core-sets for the \emph{determinant maximization} problem.
This can also be cast as the MAP inference task for \emph{determinantal point processes}, that have recently gained a lot of interest for modeling diversity and fairness.
The problem was recently studied in \cite{indyk2018composable}, where they designed composable core-sets with the optimal approximation bound of $\tilde O(k)^k$. On the other hand, the more practical \emph{Greedy} algorithm has been previously used in similar contexts. In this work, first we provide a theoretical approximation guarantee of $O(C^{k^2})$ for the Greedy algorithm in the context of composable core-sets; Further, we propose to use a \emph{Local Search} based algorithm that while being still practical, achieves a nearly optimal approximation bound of $O(k)^{2k}$;
Finally, we implement all three algorithms and show the effectiveness of our proposed algorithm on standard data sets.
\end{abstract}

\section{Introduction}\label{sec:intro}
Given a set of vectors $P\subset \Re^d$ and an integer $1\leq k\leq d$, the goal of the determinant maximization problem is to find a subset $S=\{v_1,\dots,v_k\}$ of $P$ such that the determinant of the Gram matrix of the vectors in $S$ is maximized. Geometrically,  this determinant is equal to the volume squared of the parallelepiped spanned by the points in $S$. 
This problem and its variants have been studied extensively over the last decade. To this date, the best  approximation factor is due to a work of Nikolov \cite{nikolov2015randomized} which gives a factor of $e^k$, and it is known that the exponential dependence on $k$ is unavoidable~\cite{civril2013exponential}.

The determinant of a subset of points is used as a measure of diversity in many applications where a small but diverse subset of objects must be selected as a representative of a larger population \cite{mirzasoleiman2017streaming,gong2014diverse, kulesza2012determinantal, chao2015large, kulesza2011learning, yao2016tweet, lee2016individualness}; recently, this has been further applied to model fairness~\cite{celis2016fair}.
The determinant maximization problem can also be rephrased as the maximum a posteriori probability (MAP) 
estimator for {\em determinantal point processes} (DPPs). DPPs are probabilistic models of diversity in which every subset of $k$ objects is assigned a probability proportional to the determinant of its corresponding Gram matrix. DPPs have found several applications in machine learning over the last few years \cite{kulesza2012determinantal,mirzasoleiman2017streaming,gong2014diverse,yao2016tweet}. 
In this context, the determinant maximization problem corresponds to the task of finding the most diverse subset of items.

Many of these applications need to handle large amounts of data and consequently the problem has been considered in massive data models of computation \cite{mirzasoleiman2017streaming,wei2014fast,pan2014parallel, mirzasoleiman2013distributed,mirzasoleiman2015distributed,mirrokni2015randomized,barbosa2015power}. One strong such model that we consider in this work, is \emph{composable core-set} \cite{IMMM-ccdcm-14} which is an efficient summary of a data set with the composability property: union of summaries of multiple data sets should provably result in a good summary for the union of the data sets. More precisely, in the context of the determinant maximization, a mapping function $c$ that maps any point set to one of its subsets is called an $\alpha$-composable core-set if it satisfies the following condition: given any integer $m$ and any collection of point sets $P_1,\cdots,P_m \subset \mathbb{R}^d$, 
$$
\text{MAXDET}_k (\bigcup_{i=1}^m c(P_i)) \geq \frac{1}{\alpha}\cdot \text{MAXDET}_k(\bigcup_{i=1}^mP_i)
$$
where we use $\text{MAXDET}_k$ to denote the optimum of the determinant maximization problem for parameter $k$. We also say $c$ is a core-set of size $t$ if for any $P \subset \Re^d$, $|c(P)| \leq t$.
 If designed for a task, composable core-sets will further imply efficient streaming and distributed algorithms for the same task. This has lead to recent interest in composable core-sets model since its introduction \cite{mirrokni2015randomized, assadi2017randomized, indyk2018composable}.

\paragraph{An almost optimal approximate composable core-set.} In \cite{indyk2018composable}, the authors designed composable core-sets of size $O(k\log k)$ with approximation guarantee of $\tilde{O}(k)^k$ for the determinant maximization problem. Moreover, they showed that the best approximation one can achieve is $\Omega(k^{k-o(k)})$ for any polynomial size core-sets, proving that their algorithm is almost optimal.
However, its complexity makes it less appealing in practice. First of all, the algorithm requires an explicit representation of the point set, which is not present for many DPP applications; a common case is that the DPP kernel is given by an oracle which returns the inner product between the points; in this setting, the algorithm needs to construct the associated gram matrix, and use  SVD decomposition to recover the point set, making the  time and memory quadratic in the size of the point-set. 
Secondly, even in the point set setting, the algorithm is not efficient for large inputs
as it requires
solving 
$O(kn)$ linear programs, where $n$ is size of the point set.

In this paper, we focus on two simple to implement algorithms which are typically exploited in practice, namely the Greedy and the Local-Search algorithms.
We study these algorithms from theoretical and experimental points of view for the composable core-set problem with respect to the determinant maximization objective, and we compare their performance with the algorithm of \cite{indyk2018composable}, which we refer to as the \LP ~algorithm.

\subsection{Our Contributions}
\paragraph{Greedy algorithm.}  
The greedy algorithm for determinant maximization proceeds in $k$ iterations and at each iteration it picks the point that maximizes the volume of the parallelepiped formed by the set of points picked so far. 
 \cite{cm-smvsm-09} has studied the approximation of the greedy algorithm in the standard setting. In the context of submodular maximization over large data sets,  variants of this algorithm have been studied \cite{mirzasoleiman2013distributed}. 
One can view the greedy algorithm as a heuristic for constructing a core-set of size $k$. To the best of our knowledge, the previous analysis of this  algorithm  does not provide any multiplicative approximation guarantee in the context of composable core-sets.\footnote{For more details, see related work.}

Our first result shows the first multiplicative approximation factor for composable core-sets on the determinant maximization objective achieved by the Greedy algorithm.

\begin{theorem}\label{thm:mainthm1}
Given a set of points $P\subset \mathbb{R}^d$, the Greedy algorithm achieves a $O(C^{k^2})$-composable core-set of size $k$ for the determinant maximization problem, where $C$ is a constant.
\end{theorem}

\paragraph{The Local Search algorithm.} Our main contribution is to propose to use the Local Search algorithm for constructing a composable core-set for the task of determinant maximization. The algorithm starts with the solution of the Greedy algorithm and at each iteration, swaps in a point that is not in the core-set with a point that is already in the core-set, as long as this operation increases the volume of the set of picked points. While still being simple, as we show, this algorithm achieves a near-optimal approximation guarantee. 

\begin{theorem}\label{thm:mainthm2}
Given a set of points $P\subset \mathbb{R}^d$, the Local Search algorithm achieves an $O(k)^{2k}$-composable core-set of size $k$ for the determinant maximization problem.
\end{theorem}

\paragraph{Directional height.} Both of our theoretical results use a modular framework: In Section \ref{sec:directional-height}, we introduce a new geometric notion defined for a point set called \emph{directional height}, which is closely related to the width of a point set defined in \cite{AHV-gavc-05}. We show that core-sets for preserving the directional height of a point set in fact provide core-sets for the determinant maximization problem. Finally, we show that running either the Greedy (Section \ref{sec:greedy}) or Local Search (Section \ref{sec:local-search}) algorithms on a point set obtain composable core-sets for its directional height.
We believe that this new notion might find applications elsewhere.

\paragraph{Experimental resutls.}
Finally, we implement all three algorithms and compare their performances on two real data sets: MNIST\cite{lecun1998gradient} data set and GENES data set, previously used in \cite{batmanghelich2014diversifying, li2015efficient} in the context of DPPs. Our empirical results show that in more than $87\%$ percent of the cases, the solution reported by the Local Search algorithm improves over the Greedy algorithm. The average improvement varies from $1\%$ to up to $23\%$ depending on the data set and the settings of other parameters such as $k$. We further show that although Local Search picks fewer points than the tight approximation algorithm of \cite{indyk2018composable} ($k$ vs. upto $O(k\log k)$), its performance is better and it runs faster.

\subsection{Related Work}
In a broader scope, determinant maximization is an instance of the (non-monotone) submodular maximization where   the logarithm of the determinant is the submodular objective function.  There is a long line of work on distributed submodular optimization and its variants \cite{chierichetti2010max, badanidiyuru2014streaming,mirzasoleiman2016distributed,kumar2015fast}. In particular, there has been several efforts to design composable core-sets in various  settings of the problem \cite{mirzasoleiman2013distributed,mirrokni2015randomized,barbosa2015power}; In \cite{mirzasoleiman2013distributed}, authors study the problem for monotone functions, and show the greedy algorithm offers a $min(m,k)$-composable core-set for the problem where $m$ is the number of parts. On the other hand, \cite{IMMM-ccdcm-14} shows that it is impossible to go beyond an approximation factor of $\Omega(\frac{\sqrt{k}}{\log k})$ with polynomial size core-sets. Moreover, \cite{mirrokni2015randomized,barbosa2015power} consider a variant of the problem where the data is randomly distributed,   and show the greedy algorithm achieves constant factor ``randomized'' composable core-sets for both monotone and non-monotone functions.   
However, one can notice that  these results can not be directly compared to the current work, as a multiplicative approximation for determinant  converts to an  additive guarantee for the corresponding submodular function.

As discussed before, the determinant is one way to measure the diversity of a set of items. Diversity maximization with respect to other measures has been also extensively studied in the literature, 
\cite{hassin1997approximation,gollapudi2009axiomatic,borodin2012max,bhaskara2016linear}. More recently, the problem has  received  more attention in distributed models of computation, and  for several diversity measures constant factor approximation algorithms have been devised \cite{zadeh2017scalable,IMMM-ccdcm-14,ceccarello2017mapreduce}. However, these notions are typically defined by considering only the pairwise dissimilarities between the items; for example, the summation of the dissimilarities over all pairs of items in a set can define its diversity. 

One can also go further, and study the problem under additional constraints, such as matroid and knapsack constraints. This has been an active line of research in the past few years, and several centralized and distributed algorithms have been designed in this context for submodular optimization \cite{mirzasoleiman2016fast,lee2009non,lee2010submodular,chekuri2015streaming} and in particular determinant maximization  \cite{ebrahimi2017subdeterminant,nikolov2016maximizing}.
\section{Preliminaries}
\label{sec:prelim}
Throughout the paper, we fix $d$ as the dimension of the ambient space and $k (k\leq d)$ as the size parameter of the determinant maximization problem. We call a subset of $\Re^d$ a point set, and use the term point or vector to refer to an element of a point set. For a set of points $S\subset \Re^d$ and a point $p\in \Re^d$, we write $S+p$ to denote the set $S\cup\{p\}$, and for a point $s\in S$, we write $S-p$ to denote the set $S\setminus\{s\}$.

Let $S$ be a point set of size $k$. We  use $\Vol(S)$ to denote the $k$-dimensional volume of the parallelepiped spanned by vectors  in $S$. Also, let 
  $M_S$ denote a $k\times d$ matrix where each row represents a point of $S$.  Then, the following relates volume to the determinant 
$\det(M_SM_S^\intercal) = \Vol^2(S).$
So the determinant maximization problem can also be phrased as \emph{volume maximization}. We use the former, but  because of  the geometric nature of the arguments, sometimes we switch to the volume notation. For any point set $P$, we  use $\maxdet_k$ to denote the optimal of determinant maximization for $P$, i.e. $\maxdet_k(P)= \max_{S} \det(M_SM_S^\intercal)$, where $S$ ranges over all subsets of size $k$. $\text{MAXVOL}_k$ is defined similarly. 

For a point set $P$, we use $\langle P \rangle$ to refer to the linear subspace spanned by the vectors in $P$. We also denote the set of all $k$-dimensional linear subspaces by $\h_k$. For a point $p$ and a subspace $\h$,  we use $\dist(p,\h)$ to show the Euclidean distance of $p$ from $\h$.

\paragraph{Greedy algorithm for volume maximization. }As pointed out before, a widely used algorithm for determinant maximization in the offline setting is a greedy algorithm which given a point set $P$ and a parameter $k$ as the input does the following: start with an empty set $\cs$. For $k$ iterations, add $\text{argmax}_{p \in P} \dist(p,\langle \cs \rangle )$ to $\cs$. The result would be a subset of size $k$ which has the following guarantee.
\begin{theorem}[\cite{cm-smvsm-09}]
	\label{thm:greedyoffline}
Let $P$ be a point set and $\cs$ be the output of the greedy algorithm on $P$. Then $\Vol(\cs)	\geq \frac{\maxvol_k(P)}{k!}$. 
\end{theorem}

\subsection{Core-sets}
Core-set is a generic term used for a \emph{small subset} of the data that represents it very well. More formally, for a given optimization task, a core-set is a mapping $c$ from any data set $P$ into one of its subsets such that the solution of the optimization over the core-set $c(P)$ approximates the solution of the optimization over the original data set $P$.  The notion was first introduced in \cite{AHV-aemp-04} and many variations of core-sets exist. In this work, we consider the notion of \emph{composable core-sets} defined in \cite{IMMM-ccdcm-14}.



\begin{definition}[$\alpha$-Composable Core-sets]
 A function $c(P)$ that maps the input set $P \subset \Re^d$ into one of its subsets is called an $\alpha$-composable core-set for a function $f \colon 2^{\Re^d}\to \Re$ if, for any collection of sets $P_1,\cdots, P_n\subset \Re^d$, we have
 $ f(C)\geq f(P)/\alpha $
 where $P=\bigcup_{i\leq n} P_i$ and $C=\bigcup_{i\leq n} c(P_i)$.
\end{definition}

For simplicity, we will often refer to the set $c(P)$ as the core-set for $P$ and use the term ``core-set function'' with respect to $c(\cdot)$. The {\em size} of $c(\cdot)$ is defined as the smallest number $t$ such that $c(P) \le t$ for all sets $P$ (assuming it exists). Unless otherwise stated, we might use the term ``core-set'' to refer to a composable core-set when clear from the context.
Our goal is to find composable core-sets for the determinant maximization problem.

\section{$k$-Directional Height}\label{sec:directional-height}
As pointed out in the introduction, we introduce a new geometric notion called directional height, and reduce the task of  finding composable core-sets for determinant maximization to finding core-sets for this new notion. 
\begin{definition}[$k$-Directional Height]
Let $P\subset \Re^d$ be a point set and $\h \in \h_{k-1}$ be a $(k-1)$-dimensional subspace. We define the $k$-directional height of $P$ with respect to $\h$, denoted by $h(P,\h)$, to be the distance of the farthest point in $P$ from $\h$, i.e., $h(P,\h) = \max_{p\in P} \dist(p,\h)$.
\end{definition}
The notion is an instance of an \emph{extent measure} defined in 
\cite{AHV-gavc-05}. It is also related to the notion of \emph{directional width} of a point 
set previously used in \cite{AHV-gavc-05}, which for a direction vector $v\in \Re^d$ is 
defined to be $\max_{p\in P} \langle v\cdot p\rangle - \min_{p\in P} \langle v\cdot 
p\rangle$.

Next, we define core-sets with respect to this notion. It is essentially a subset of the point set that approximately preserves the $k$-directional height of the point set with respect  to any subspace in $\h_k$.
\begin{definition}[$\alpha$-Approximate Core-set for the $k$-Directional Height]
\label{d:height}
Given a point set $P$, a subset $C\subseteq P$ is a $\alpha$-approximate core-set for the $k$-directional height of $P$, if for any $\h \in \h_{k-1}$, we have $h(C,\h)\geq h(P,\h)/\alpha$.

We also say a mapping $c(.)$ which maps any point set in $\Re^d$ to one of its subsets, is an $\alpha$-approximate core-set for the $k$-directional height problem, if the above relation holds for any point set $P$ and $c(P)$.
\end{definition}
The above notion of core-sets for $k$-directional height is similar to the notion of $\eps$-kernels defined in \cite{AHV-gavc-05} for the directional width of a point set. 

We connect it to composable core-sets for determinant maximization by the following lemma. 
\begin{lemma}
	
	Let $P_1,\dots,P_m \in \Re^d$ be an arbitrary collection of point sets, and for any $i$, let $c(P_i)$ be an $\alpha$-approximate core-set for the $k$-directional height for $P_i$. Then 
	$$\maxdet_k (\bigcup_{i=1}^m P_i)\leq \alpha^{2k}\cdot \maxdet_k (\bigcup_{i=1}^m c(P_i) ).$$ 
\end{lemma}
\begin{proof}
	Let $W \subset \bigcup_{i=1}^m P_i$ be any subset of size $k$, and also let $w \in W \setminus \bigcup_{i=1}^m c(P_i)$. We claim that there is a point $q$ in the union of the core-sets such that $\alpha\cdot \Vol(W-w+q) \geq \Vol(W)$.  Note that showing this claim is enough  to prove the lemma. Since, one can start from the optimum solution which achieves the largest volume on $\bigcup_{i=1}^m P_i$, and for at most $k$ iterations,  replace a point outside  $\bigcup_{i=1}^m c(P_i)$ by a point inside, while decreasing  the volume by a factor of at most $\alpha$. 
	
	So it remains to prove the claim. Let $W=\{w_1,\dots,w_k\}$, and let $H=\langle w_2,\dots,w_k \rangle \in \mathcal{H}_{k-1}$ be the plane spanned by $w_2,\dots,w_k$. By definition, $\Vol(W)=\dist(w_1,H)\cdot \Vol(W-w_1)$. On the other hand, suppose that $w_1\in P_i$. Then by our assumption, there exists $q \in c(P_i)$ so that $\dist(q,H) \geq \frac{\dist(w_1,H)}{\alpha}$. Replacing $w_1$ with $q$, we get
	\begin{align*}
	&\Vol(W-w_1+q) = \dist(q,H)\cdot \Vol(W-w_1)\\&\geq \frac{\dist(w_1,H)\cdot \Vol(W-w_1)}{\alpha} = \frac{\Vol(W)}{\alpha}
	\end{align*}
	which completes the proof.  
\end{proof}

\begin{corollary}
	\label{cor:dirtodet}
	Any mapping which is an $\alpha$-approximate core-set for $k$-directional height, is an $\alpha^{2k}$-composable core-set for the determinant maximization.
\end{corollary}

We employ the above corollary to analyze both greedy and local search algorithms in Sections \ref{sec:local-search} and \ref{sec:greedy}.

\section{The Local Search Algorithm}\label{sec:local-search}
\label{sec:localsearch}
In this section, we describe and analyze the local search algorithm and prove \autoref{thm:mainthm2}. 
The algorithm  is  described in Algorithm \ref{alg:LS}. 

To prove \autoref{thm:mainthm2}, we follow a two steps strategy.  We first analyze the algorithm for individual point sets, and show that the output is a $(2k)$-approximate core-set for the $k$-directional height problem, as follows.
\begin{lemma}\label{lem:lcs-coreset}
	Let $\pntset$ be a set of points and $c(P)$ be the result of running the local search algorithm on $\pntset$. Then, for any $\h \in \mathcal{H}_{k-1}$, 
\[
h(c(P),\h)\geq \frac{h(P,\h)}{2k(1+\epsilon)}.
\]
\end{lemma}
Next, we apply \autoref{cor:dirtodet}, which implies that local search gives $(2k(1+\epsilon))^{2k}$-composable core-sets for the  determinant maximization. Clearly this completes the proof of the theorem by setting $\epsilon$ to a constant.
 
So proving \autoref{thm:mainthm2} boils down to showing \autoref{lem:lcs-coreset}. Before, getting into that, we analyze the running time, and present some remarks about the implementation.

\paragraph{Running time.} Let $\mathcal{C}_0$ be the output of the greedy. By \autoref{thm:greedyoffline}  $\frac{\Vol(\mathcal{C}_0)}{\text{MAXVOL}_k(P)} \geq \frac{1}{k!}$. The algorithm starts with $\mathcal{C}_0$ and by definition, in any iteration increases the volume by a factor of at least $1+\epsilon$, hence the total number of iterations is $O(\frac{k\log k}{\log(1+\epsilon)})$. Finally, each iteration can be naively executed by iterating over all points in $P$, forming the corresponding $k\times k$ matrix, and computing the determinant in total time $O(|P|\cdot kd\cdot k^3|P|)$.
  
We also remark that unlike the algorithm in \cite{indyk2018composable}, our method can also be executed without any changes and additional complexity, when the point set $P$ is not explicitly given in the input; instead, it is presented by an oracle  that given two points of $P$ returns their inner product. One can note that in this case the algorithm can be simulated by querying this oracle for at most $O(|P|k)$ times.


\begin{algorithm}[tb]
	\caption{Local Search Algorithm}
	\label{alg:LS}
	\begin{algorithmic}
		\STATE {\bfseries Input:} A point set $P\subset \Re^d$, integer $k$, and  $\epsilon>0$.
		\STATE {\bfseries Output:} A set $\mathcal{C} \subset P$ of size $k$.
		\STATE Initialize $\mathcal{C}=\emptyset$. 
		\FOR{$i=1$ {\bfseries to} $k$}
		\STATE Add $\text{argmax}_{p \in P\setminus{\mathcal{C}}} \Vol(\mathcal{C}+ p)$ to $\mathcal{C}$.
		\ENDFOR
		\REPEAT
		\STATE If there are points $q \in P\setminus{C}$ and $p \in \mathcal{C}$ such that
		\[\Vol(\mathcal{C}+q-p) \geq (1+\epsilon)\Vol(\mathcal{C})\] replace $p$ with $q$.
		\UNTIL{No such pair exists.}
		\STATE {\bfseries Return} $\mathcal{C}$.
	\end{algorithmic}
\end{algorithm}
\vspace{1mm}


\subsection{Proof of \autoref{lem:lcs-coreset}}
With no loss of generality, suppose that $\epsilon=0$, the proof automatically extends to $\epsilon\neq 0$. Therefore, $c(P)$ has the following property: for any $v \in P\setminus c(P)$ and $u \in c(P)$, $\Vol(c(P))\geq \Vol(c(P)-u+v)$. Fix $\h \in \h_{k-1}$, and let $p= \text{argmax}_{p\in P} \dist(p,\h)$. Our goal is to show there exists $q \in c(P)$ so that $\dist(q,\mathcal{H})\geq \frac{\dist(p,\mathcal{H})}{2k}$. 

Let $\ha=\langle c(P) \rangle $ be  the $k$-dimensional linear subspace spanned by the set of points in the core-set, and let $p_\ha$ be the projection of $p$ onto this subspace.
We proceed with proof by contradiction. Set $\dist(p,\mathcal{H})=2x$, and suppose to the contrary that for any $q \in c(P)$, $\dist(q,\h) < \frac{x}{k}$.
  With this assumption, we prove the two following lemmas.
 
 \begin{lemma}\label{lem:lcs-first}
 	$\dist(p,p_{\ha})<x$.
 \end{lemma}

\begin{lemma}\label{lem:lcs-second}
	$\dist(p_{\ha},\h)< x$.
\end{lemma}
One can note that, combining the above two lemmas and applying the triangle inequality implies 
$\dist(p,\h)\leq \dist(p,p_{\ha})+ \dist(p_{\ha},\h)<2x$, which contradicts the assumption $\dist(p,\h)=2x$ and completes the proof.  

Therefore, it only remains to prove the above lemmas. Let us first fix some notation. 
Let $c(P)=\{q_1,\dots,q_k\}$ and for any $i$, let $\ha_{\bar i}$ denote the $(k-1)$-dimensional subspace spanned by points in $c(P)\setminus \Set{q_i}$. 
 

\begin{proofof}{\autoref{lem:lcs-first}}
For $1\leq i \leq k$, let $q_i'$ be the projection of $q_i$ onto $\h$. 
We prove that there exists an index $j\leq k$ such that we can write $q_j' = \sum_{i\neq j} \alpha_i q_i'$ where every $\alpha_i\leq 1$.  Let $r$ be the rank, i.e., maximum number of independent points of $\cs'  = \Set{q_i'|i\leq k}$ and clearly as $\h$ has dimension $k-1$, we have $r\leq k-1$. 
Take a subset $S\subset \cs'$ of $r$ independent points that have the maximum volume and let $q_j'$ be a point in $\cs'\setminus S$ and note that this point should exist as there are $k$ points in the core-set. Thus we can write $q_j' =\sum_{i: q_i'\in S} \alpha_i q_i'$. With an idea similar to the one presented in \cite{cm-smvsm-09}, we can prove that the following claim holds.
\begin{claim}\label{clm:lcs-maxvol}
For any $i$ such that $q_i'\in S$, we have $\abs{\alpha_i} \leq 1$.
\end{claim}
\begin{proof}
We prove that if the claim is not true, then $S\setminus \Set{q_i'} \cup \Set{q_j'}$ has a larger volume than $S$ which contradicts the choice of $S$. Let $\mathcal{F}$  be the linear subspace passing through $S\setminus \Set{q_i'}$. It is easy to see that $\frac{\vol(S)}{\vol(S\setminus \Set{q_i'} \cup \Set{q_j'})} = \frac{\dist(q_i',\mathcal{F})}{\dist(q_j',\mathcal{F})}$, meaning that $\dist(q_i',\mathcal{F}) \geq \dist(q_j',\mathcal{F})$. However, if $\abs{\alpha_i}>1$ then since $q_i'$ is the only point in $S$ which in not in $\mathcal{F}$, then $\dist(q_j',\mathcal{F}) \geq \dist(q_i',\mathcal{F})$ which is a contradiction.
\end{proof}

Finally, for any $q_i' \notin S$, set the corresponding coefficient $\alpha_i = 0$. So we get that $q_j' = \sum_{i\neq j} \alpha_i q_i'$ where every $\abs{\alpha_i}\leq 1$.

Now take the point $q=\sum_{i\neq j} \alpha_i q_i$. 
Note that,  $q_j'$ is in fact the projection of $q$ onto $\h$. Therefore,  
 using triangle inequality, we have 
\begin{equation}
\begin{aligned}\label{eq:lcs-qkq}
\dist(q_j',q) = \dist(q,\h) \leq \sum_{i\neq j} \abs{\alpha_i} \dist(q_i,\h)  \leq \frac{(k-1) x}{k}
\end{aligned}
\end{equation}
Then we get that
\[
\begin{split}
&\dist(p,p_{\ha}) =\dist(p,\ha) \\
&\leq \dist(p,\ha_{\bar j}) \quad \explain{as $\ha_{\bar j}\subset \ha$} \\
& \leq \dist(q_j,\ha_{\bar j}) \quad \explain{by the local search property} \\
& \leq \dist(q_j,q) \quad \explain{as $q \in \ha_{\bar j}$} \\
& \leq \dist(q_j,q_j') + \dist(q_j',q) \quad \explain{by triangle inequality} \\
& < x/k + (k-1) x/k   \quad \explain{by our assumption and Equation \ref{eq:lcs-qkq}}\\
&= x
\end{split}
\] \end{proofof}

\begin{proofof}{\autoref{lem:lcs-second}}
Again we prove that we can write $p_{\ha}=\sum_{i=1}^k \alpha_i q_i$ where all $\abs{\alpha_i} \leq 1$. We assume that the set of points $q_i$ are linearly independent, otherwise the points in $\pntset$ have rank less than $k$ and thus the volume is $0$.
Therefore, we can write $p_{\ha} = \sum_{i=1}^k \alpha_i q_i$. Note that for any $i$, we have
\begin{align*}
\dist(p_{\ha},\ha_{\bar i}) &\leq \dist(p,\ha_{\bar i}) \\
& \leq \dist(q_i,\ha_{\bar i}) \quad \explain{by the local search property} 
\end{align*}
where the first inequality follows since $\ha_{\bar i}$ is a subspace of $\ha$ and $p_{\ha}$ is the projection of $p$ onto $\ha$. Again, similar to the proof of Claim \ref{clm:lcs-maxvol}, this means that $\abs{\alpha_i}\leq 1$.
Therefore, using triangle inequality
\[
\begin{split}
\dist(p_{\ha},\h) &= \dist(\sum_{i=1}^k \alpha_i q_i, \h) \leq \sum_{i=1}^k \abs{\alpha_i}\dist(q_i,\h) \\
&< k \times x/k = x.
\end{split}
\]
\end{proofof}

\section{The Greedy Algorithm}\label{sec:greedy}
In this section we analyze the performance of the greedy algorithm (see \autoref{sec:prelim}) as a composable core-set function for the determinant maximization and prove \autoref{thm:mainthm1}. Our proof plan is similar to the to the analysis of the local search. We analyze the guarantee of the greedy as a core-set mapping for $k$-directional height, and  combining that with \autoref{cor:dirtodet} we achieve the result. We prove the following.

\begin{lemma}\label{lem:greedy-coreset}
	Let $\pntset$ be an arbitrary point set and $c(P)$ denote the output of running greedy on $\pntset$.  Then, $c(P)$ is a $(2k)\cdot 3^k$-approximate core-set for the $k$-directional height of $\pntset$, i.e. for any $\h \in \h_{k-1}$ we have 
	$$h(c(P),\h) \geq \frac{1}{2k\cdot 3^k}\cdot h(P,\h)$$
\end{lemma}
So the greedy is  a $(2k\cdot 3^k)$-approximate core-set for $k$-directional height problem. Combining with \autoref{cor:dirtodet}, we conclude it is also a $(2k\cdot 3^k)^{2k}$ composable core-set for the determinant maximization which proves \autoref{thm:mainthm1}. 

\subsection{Proof of \autoref{lem:greedy-coreset}}
 The proof is similar to the proof of \autoref{lem:lcs-coreset}. 
Let $\ha=\langle c(P) \rangle$ be the $k$-dimensional subspace spanned by the output of greedy. Also for a point $p \in \pntset$, define $p_\ha$ to be its projection onto $\ha$. 
Fix $\h \in \h_{k-1}$, let 
$h(c(P),\h)=\frac{x}{k}$ for some number $x$, which in particular implies that for any $q \in c(P)$, $\dist(q,\h)\leq \frac{x}{k}$. Then, our goal is to prove $h(\pntset,\h)\leq 2\cdot 3^k\cdot x$.  We show that by proving  the following two lemmas. %
 
 \begin{lemma}\label{lem:d-p'-A}
 	For any $p \in P$,  $\dist(p_{\ha},\h)\leq 2^{k-1} x$.
 \end{lemma}

\begin{lemma} \label{lem:greedy-d-p-p'}
	For any $p \in \pntset$, $\dist(p,p_{\ha})\leq 3^k x$.
\end{lemma}
Clearly, combining them with triangle inequality, we get that for any $p \in \pntset$, $\dist(p,\h) <2\cdot 3^k x$, which implies $h(\pntset,\h)\leq 2\cdot 3^k\cdot x$ and completes the proof.  So it remains to prove the lemmas. Let the  output of the greedy $c(P)$ be $q_1, \dots,q_k$ with this order, i.e. $q_1$ is the first vector selected by the algorithm. 

\begin{proofof}{\autoref{lem:d-p'-A}}
	Recall that $q_1,\dots,q_k$ is the output of greedy. For any $p\in \pntset$ and for any $1\leq t \leq k$, let $\ha_t=\langle q_1,\dots,q_t \rangle$ and  define $p^t$ to be the projection of $p$ onto $\ha_t$. We show the lemma using the following claim.
	\begin{claim}\label{claim:greedycoreset}
		For any $p\in \pntset$ and any $1\leq t \leq k$,
		we can write $p^t=\sum_{i=1}^t \alpha_i q_i$ so that for each $i$, $|\alpha_i|\leq 2^{t-1}$.
	\end{claim}
Let us first show how  the above claim implies the lemma. It follows that  we can write $p_{\ha}=p^k=\sum_{i=1}^k \alpha_i q_i$ where all $\abs{\alpha_i} \leq 2^{k-1}$. Now since for each $i\leq k$, $\dist(q_i,\h)\leq x/k$ by assumption, we have that $\dist(p_{\ha},\h)\leq \sum\alpha_i \dist(q_i,\h) \leq 2^{k-1}x$.
Therefore, it suffices to prove the claim.

\begin{proofof}{\autoref{claim:greedycoreset}}
	We use induction on $t$. 
	To prove the base case of induction, i.e., $t=1$, note that $q_1$ is the vector with largest norm in $\pntset$. Thus we have that $\norm{p^1}\leq \norm{q_1}$ and therefore we can write $p^1 = \alpha_1 q_1$ where $\abs{\alpha_1}\leq 1$.  Now, lets assume that the hypothesis holds for the first $t$ points; that is, the projection of any point $p$ onto $\ha_t$ can be written as $\sum_{j\leq t} \alpha_j q_j$ where $\abs{\alpha_j}$'s are at most $2^{t-1}$. 
	
	Now, note that by the definition of the greedy algorithm, $q_{t+1}$ is the point with farthest distance from $\ha_t$. Therefore, for any point $p\in \pntset \setminus \Set{q_1,\cdots,q_{t+1}}$, we know that $\dist(p,\ha_t)\leq \dist(q_{t+1},\ha_t)$, and thus, $\dist(p^{t+1},\ha_t)\leq \dist(q_{t+1},\ha_t)$. Therefore if we define $q_{t+1}^t$ to be the projection of $q_{t+1}$ onto $\ha_t$, we can write 
	\[
	p^{t+1} = \alpha_{t+1}q_{t+1} - \alpha_{t+1} q_{t+1}^{t}+ p^{t} \quad \mbox{ where } \abs{\alpha_{t+1}} \leq 1.
	\]
	By the hypothesis, we can write $p^{t}=\sum_{j\leq t} \beta_j q_j$, and $q_{t+1}^t=\sum_{j\leq t} \gamma_j q_j$, where $\abs{\beta_j}\leq 2^{t-1}$, and $\abs{\gamma_j}\leq 2^{t-1}$. Since $\abs{\alpha_{t+1}}\leq 1$, we can write 
	\[
	p^{t+1} = \alpha_{t+1}q_{t+1} + \sum_{j\leq t} (\beta_j - \alpha_{t+1}\gamma_j) q_j =  \sum_{j\leq t+1} \alpha_j q_j\] 
	where  $\abs{\alpha_{j}} \leq 2^t$. This completes the proof.
\end{proofof}

\end{proofof}
\vspace{-0.2cm}
\begin{proofof}{\autoref{lem:greedy-d-p-p'}}
First, note that for any $t$, we have 
$\dist(q_{t+1}, \ha_{t}) \geq \dist(p,\ha_{k-1}).$
This is because the greedy algorithm 
has chosen $q_k$ over $p$ in its $k$-th round which means that 
$\dist(p,\ha_{k-1})\leq \dist(q_k,\ha_{k-1})$, and by definition of the greedy algorithm for any $i <j$ we have 
$\dist(q_{i+1},\ha_{i})\geq \dist(q_{j+1},\ha_{j})$. So it is enough to prove 
\begin{equation}
\label{eq:goal}
\exists \, 1 \leq t \leq k-1 \text{ s.t. } \dist(q_{t+1},\ha_{t})\leq 
3^{k}x
\end{equation}
For $1\leq i \leq k$, let $q_i'$ be the projection of $q_i$ onto 
$\h$. Recall that, we are assuming 
that for any $i$, $\dist(q_i,q_i')< x/k$. To prove \eqref{eq:goal}, we use proof by contradiction, so suppose that for all $t$, $\dist(q_{t+1},\ha_{t}) > 3^kx$. We also  define $\ha_t'$ to be the projection of $\ha_t$ on $\h$, i.e., $\ha'_t=\langle q'_1,\dots,q'_t \rangle$. 
Given these assumptions, we prove the following claim.	
\begin{claim}\label{clm:maxvol}
	For any $1 \leq t \leq k-1$, we can write $\Pi(\ha'_t)(q_{t+1}') = \sum_{i\leq t} \alpha_i q_i'$ where $\abs{\alpha_i}\leq 3^t$, where for a point $q$ and a subspace $\mathcal{A}$, $\Pi(\mathcal{A})(q)$ denotes projection of $q$ onto $\mathcal{A}$.
\end{claim}



\begin{proof}
Intuitively, this is similar to  \autoref{claim:greedycoreset}.
However, instead of looking at the execution of the algorithm on the points $q_1,\cdots, q_k$, we look at the execution of the algorithm on the projected points $q_1',\cdots,q_k'$. Since all of these $k$ points are relatively close to the hyperplane $\h$, the distances are not distorted by much and therefore, we can get approximately the same bounds. 
	 Formally, we prove the claim by induction on $t$, and show that for any $j$ s.t. $j > t$, the point $\Pi(\ha_t')(q_j')$ can be written as the sum $\sum_{i\leq t} \alpha_i q_i'$ such that $\abs{\alpha_i}\leq 3^t$.
	
	\smallskip\noindent\textbf{Base Case.} First, we prove the base case of induction, i.e., $t=1$.  Recall that by our assumption, 
	$\norm{q_1}> 3^k x$, and thus by triangle inequality,
	we have that $\norm{q_1'}\geq \norm{q_1}-x/k \geq 3^k x - x/k\geq 2x$.
	Therefore, since $q_1$ is the vector with largest norm in $\pntset$, using triangle inequality again, we have that for any $j>1$,
	\[
	\norm{q_j'}\leq \norm{q_j} \leq \norm{q_1} \leq \norm{q_1'}+x/k \leq (1+\frac{1}{2k})\norm{q_1'}
	\]
	 Therefore we can write $\Pi(\ha_1')(q_j') = \alpha_1 q_1'$ where $\abs{\alpha_1}\leq 2$. 
	
	\smallskip\noindent\textbf{Inductive step.} Now, lets assume that the hypothesis holds for $\ha_{t}'$. In particular this means that we can write $\Pi(\ha_t')(q_{t+1}') = \sum_{i\leq t}\beta_i q_i'$ where $\abs{\beta_i}\leq 3^t$, and that for a given $j>t+1$, we can write $\Pi(\ha_t')(q_{j}') = \sum_{i\leq t} \gamma_i q_i'$ where $\abs{\gamma_i}$'s are at most $3^{t}$.
	Now let $\ell = \dist(q_{t+1}',\ha_t')$. By triangle inequality, we get that
	\begin{align}\label{eq:greedy-triangle}
	\dist(q_{t+1},\ha_t) &\leq \dist(q_{t+1},q_{t+1}') \\
	&+ \dist(q_{t+1}',\Pi(\ha_{t}')(q_{t+1}') + \\
	&\dist(\Pi(\ha_{t}')(q_{t+1}'),\ha_t) \nonumber \\
	&\leq x/k + \ell+ \dist(\sum_{i\leq t} \beta_i q_i' , \sum_{i\leq t} \beta_i q_i) \nonumber \\
	&\leq  x/k + \ell+ \sum_{i\leq t} \abs{\beta_i} x/k \nonumber \\
	&\leq \ell + 3^t x.
	\end{align}
	Now we consider two case. If $\ell \leq 3^tx$ then 
	using the above 
	\[
	\dist(q_{t+1},\ha_t) \leq 2 \cdot 3^tx \leq 3^k x,
	\]
	which contradicts our assumption of $\dist(q_{t+1},\ha_t)>3^k x$.
	Otherwise, 
	\[
	\begin{split}
	\dist(\Pi(\ha_{t+1}')(q_j'), \ha_{t}') &\leq \dist(q_j',\ha_t') \leq \dist(q_j,\ha_t) \\
	&\leq \dist(q_{t+1},\ha_t) \leq 2\ell,
	\end{split}
	\] 
	where the last inequality follows from Equation \ref{eq:greedy-triangle} . Therefore, we can write $\Pi(\ha_{t+1}')(q_j') = \alpha_{t+1} q_{t+1}' - \alpha_{t+1} \Pi(\ha_{t}')(q_{t+1}') + \Pi(\ha_{t})(q_j')$ where $\alpha_{t+1}\leq 2$.
	
	By the hypothesis, we can write $\Pi(\ha_t')(q_j')=\sum_{i\leq t} \gamma_i q_i'$, where $\abs{\gamma_i}\leq 3^{t}$. Since $\abs{\alpha_{t+1}}\leq 2$, we can write 
	\[
	\begin{split}
	\Pi(\ha'_{t+1})(q_j') &= \alpha_{t+1}q_{t+1}' + \sum_{i\leq t} (\gamma_i - \alpha_{t+1}\beta_i) q_i' \\
	&=  \sum_{i\leq t+1} \alpha_i q_i' \quad \mbox{ where } \abs{\alpha_{i}} \leq 3^{t+1}.
	\end{split}
	\]
	This completes the proof of the claim.
\end{proof}
To finish the proof of the lemma, let us show how it follows from the claim. First, note that $q_1',\dots,q_k'$ are $k$ points in the $(k-1)$-dimensional space $\h$, so for some $t$, $q'_{t+1}$ should lie inside $\ha'_t$ and we have $\Pi(\ha'_t)(q'_{t+1}) = q'_{t+1}$. Fix such $t$. Define the point $q_\alpha = \sum_{i\leq t} \alpha_i q_i$ where $\abs{\alpha_i}\leq 3^k$ are taken from the above claim which means $q'_{t+1}=\sum_{i\leq t} \alpha_i q'_i.$
Note that by definition $q'_{t+1} = \Pi(\h)(q_{\alpha})$.
Therefore, 
\begin{align}\label{eq:qkq}
\dist(q_{t+1}',q_\alpha) &=\dist(q_{\alpha},\h) \\ &\leq \sum_{i\leq t} \alpha_i \dist(q_i,\h) \leq 3^k t\cdot  x/k.
\end{align}
Then we get that
\[
\begin{split}
 \dist(q_{t+1},\ha_{t}) & \leq \dist(q_{t+1},q_\alpha) \quad \explain{as $q_{\alpha} \in \ha_{t}$} \\
& \leq \dist(q_{t+1},q_{t+1}') + \dist(q_{t+1}',q_\alpha) 
\\
& \leq x/k + 3^kt\cdot  x/k \leq 3^k x  
\end{split}
\]
where the second inequality holds because of triangle inequality and the last one from \eqref{eq:qkq} and the fact that $t\leq k-1$. This contradicts our assumption that $\dist(q_{t+1},\ha_t) >3^k x$, and proves the lemma.
\end{proofof}

\section{Experiments}\label{sec:experiments}
In this section, we evaluate the effectiveness of our proposed Local Search algorithm empirically on real data sets. We implement the following three algorithms.

\begin{itemize}
\vspace{-0.2cm}
\item{The Greedy algorithm of Section \ref{sec:greedy}} (GD).
\vspace{-0.2cm}
\item{The Local Search algorithm of Section \ref{sec:local-search} with accuracy parameter $\eps = 10^{-5}$} (LS).
\vspace{-0.2cm}
\item{The LP-based algorithm of \cite{indyk2018composable} which has almost tight approximation guarantee theoretically (LP). Note that this algorithm might pick up to $O(k\log k)$ points in the core-set.}
\end{itemize}

\paragraph{Data sets.}
We use two data sets that were also used in \cite{li2015efficient} in the context of approximating DPPs over large data sets.
\begin{itemize}
\vspace{-0.3cm}
\item MNIST \cite{lecun1998gradient}: contains a set of $60000$ images of hand-written digits, where each image is of size $28$ by $28$. 
\vspace{-0.1cm}
\item
GENES \cite{batmanghelich2014diversifying}: contains a set of $10000$ genes, where each entry is a feature vector of a gene. The features correspond to shortest path distances of $330$ different hubs in the BioGRID gene interaction network. This data set was initially used to identify a diverse set of genes to predict a tumor.  Here, we slightly modify it and remove genes that have an unknown value at any coordinate which gives us a data set of size $\sim 8000$.
\end{itemize}
\vspace{-0.3cm}
Moreover, we apply an RBF kernel on both of these data sets using $\sigma=6$ for MNIST and $\sigma=10$ for GENES. These are the same values used in the work of \cite{li2015efficient}.

\subsection{Experiment Setup.}
We partition the data sets uniformly at random into multiple data sets $P_1,\cdots, P_m$. We use  $m=10$ for the smaller GENES data set, and for the larger MNIST data set we use $m=50$ and also we use $m=10$ (equal to the number of digits in the data set).  Moreover, since the partitions are random, we repeat every experiment $10$ times and take the average in our reported results.

We then use a \emph{core-set construction algorithm} $\alg_c$ to compute core-sets of size $k$, i.e., $S_1=\alg_c(P_1,k),$$\cdots, $$S_m= \alg_c(P_m,k)$, for $\alg_c\in \{\mbox{GD, LS, LP}\}$.  Recall that GD, LS and LP  correspond to the Greedy, Local Search and LP-based algorithm of \cite{indyk2018composable} respectively.

Finally, we take the union of these core-sets $U_{\alg_c} = S_1\cup\cdots\cup S_m$ and  compute the solutions for $U_{\alg_c}$. Since computing the optimal solution can take exponential time ($\sim n^k$), we will instead use an \emph{aggregation algorithm} $\alg_a$ (either GD, LS or LP).  We will use the notation $\alg_a/\alg_c$ to refer to the constructed set of $k$ points, returned by $\alg_a(U_{\alg_c},k)$. For example, GD/LS refers to the set of $k$ points returned by the Greedy algorithm on the union of the core-sets, where each core-set is produced using the Local Search algorithm.


Finally, we vary the value of $k$ from $3$ to $20$.
\subsection{Results}
\paragraph{Local Search vs. Greedy as offline algorithms.} 
Our first set of experiments simply compares the quality of Greedy and Local Search as  centralized algorithms 
on whole data sets. We perform this experiment to measure the improvement of Local Search over Greedy in the offline setting. Intuitively, this improvement upper bounds the improvement one can expect in the core-set setting.
Figure \ref{fig:lsgd-whole-ratio} shows the improvement ratio of the determinant of the solution returned by the Local Search algorithm over the determinant of the solution returned by the Greedy algorithm. On average over all values of $k$, Local Search improves over Greedy by $13\%$ for GENES data set and $5\%$ for MNIST data set. Figure \ref{fig:lsgd-whole-timeratio} shows the ratio of the time it takes to run the Local Search and Greedy algorithms as a function of $k$ for both data sets. On average, it takes about $6.5$ times more to run the Local Search algorithm.

\begin{figure}[!h]
\minipage{0.45\textwidth}
  \includegraphics[width=\linewidth]{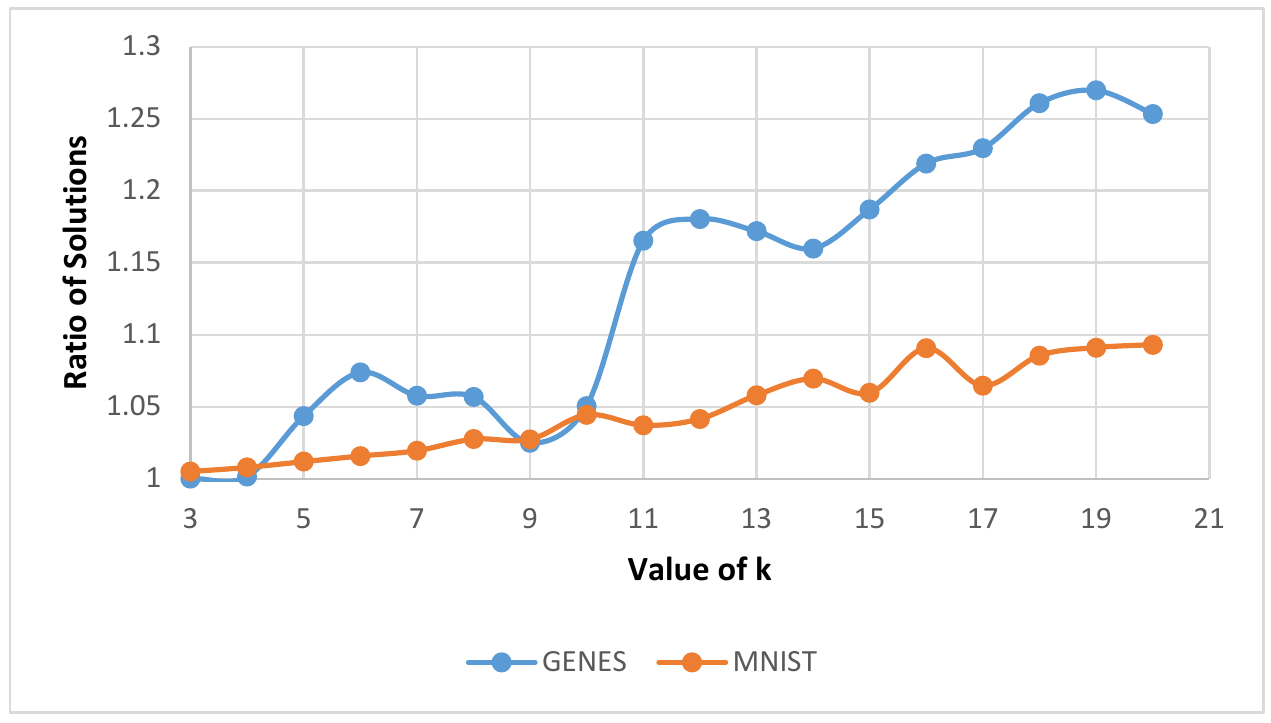}
\caption{Average improvement of Local Search over Greedy as a function of $k$.}
\label{fig:lsgd-whole-ratio}
\endminipage\hfill
\minipage{0.45\textwidth}
 \includegraphics[width=\linewidth]{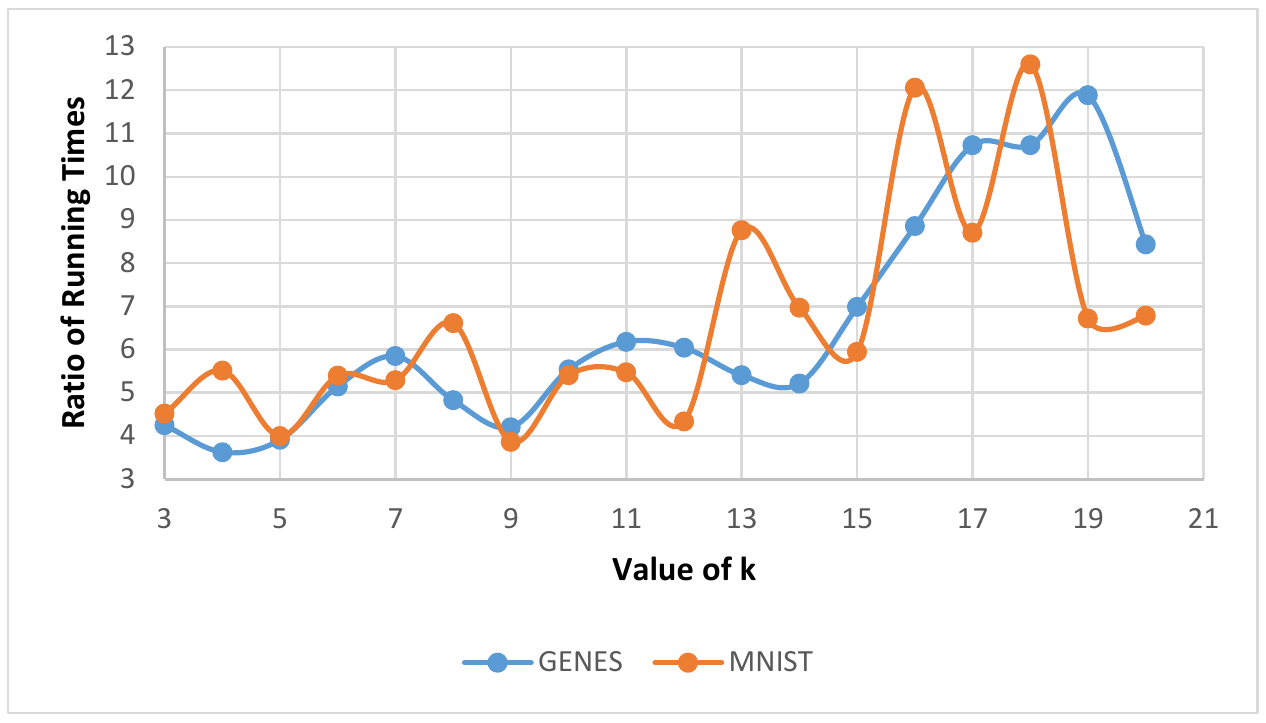}
 \caption{Average ratio of the run time of Local Search over Greedy as a function of $k$.}
 \label{fig:lsgd-whole-timeratio}
\endminipage
\end{figure}

\paragraph{Local Search vs. Greedy as core-sets.}
In our second experiment, we use  Greedy algorithm for aggregation, i.e., $\alg_a=GD$, and compare GD/LS with GD/GD.
Figure \ref{fig:lsgd-coreset-ratio} shows the improvement of local search over greedy as a core-set construction algorithm. The graph is drawn as a function of $k$, and for each $k$, the improvement ratio is an average over all 10 runs, and shown for all data sets (including GENES, MNIST with partition number $m=10$, and MNIST with $m=50$).

On average this improvement is $9.6\%$, $2.5\%$ and $1.9\%$ for GENES, MNIST10 and MNIST50 respectively. Moreover, in $87\%$ of all 180 runs of this experiment, Local Search performed better than Greedy, and for some instances, this improvement was up to $58\%$.
Finally, this improvement comes at a cost of increased running time. Figure \ref{fig:lsgd-coreset-timeratio} shows average ratio of the time to construct core-sets using Local Search vs. Greedy.

\begin{figure}[!h]
\minipage{0.45\textwidth}
  \includegraphics[width=\linewidth]{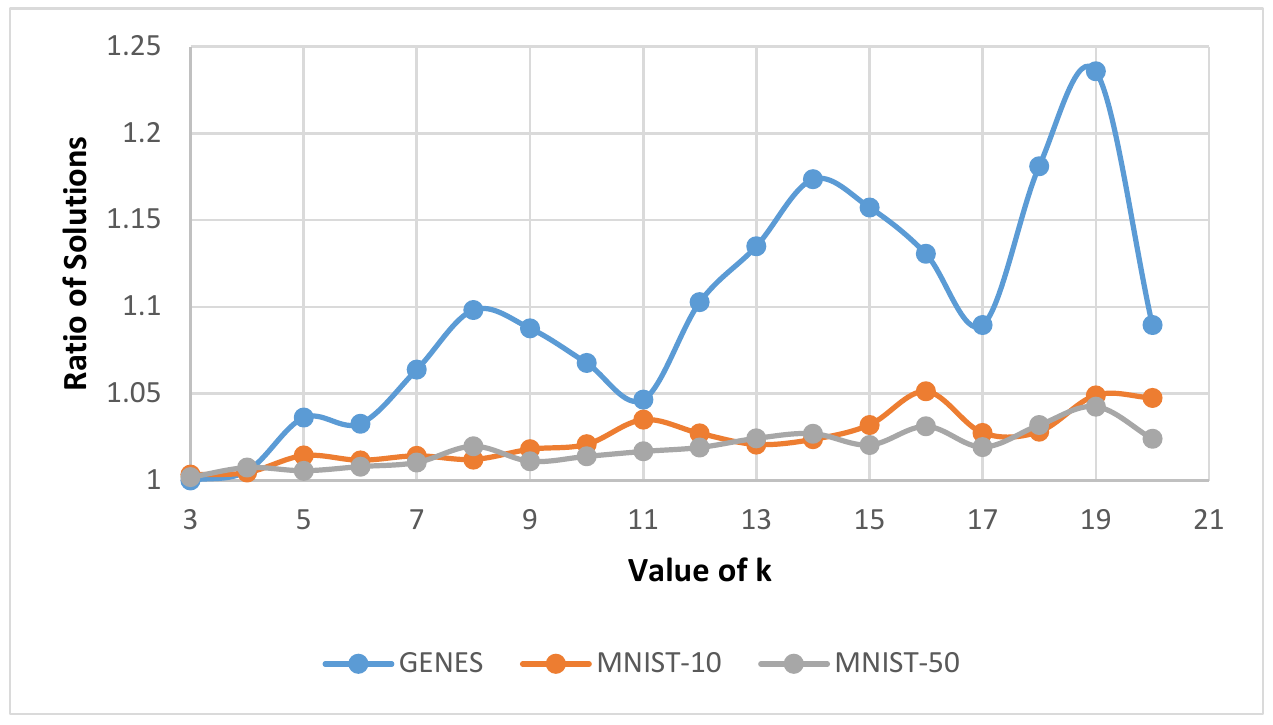}
\caption{Average improvement of Local Search core-set over Greedy core-set as a function of $k$.}
\label{fig:lsgd-coreset-ratio}
\endminipage\hfill
\minipage{0.45\textwidth}
 \includegraphics[width=\linewidth]{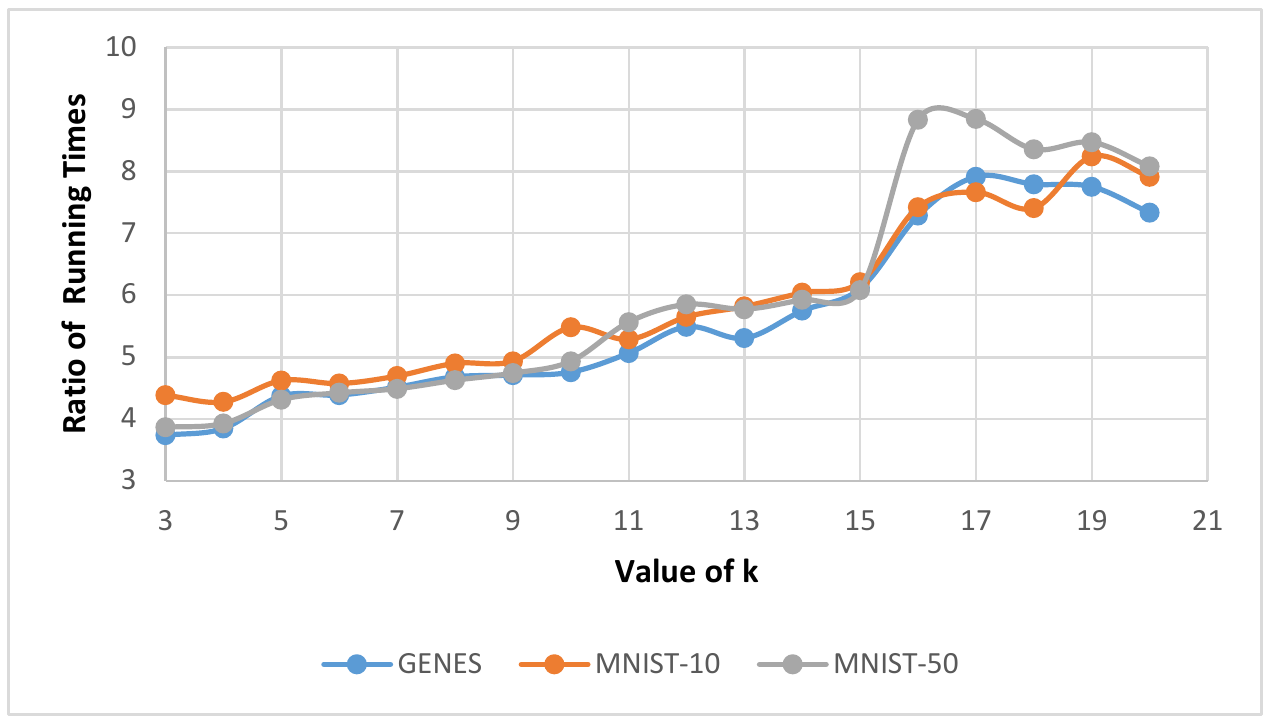}
\caption{Average ratio of the run time of Local Search over Greedy as a function of $k$.}
\label{fig:lsgd-coreset-timeratio}
\endminipage
\end{figure}

\paragraph{Local Search vs. Greedy - identical algorithms.}
We also consider the setting where the core-set construction algorithm is the same as the aggregation algorithm.
This mimics the approach of \cite{mirzasoleiman2013distributed},
who proposed to use the greedy algorithm on each machine to achieve a small solution; then each machine sends this solution to a single machine that further runs the greedy algorithm on the union of these solutions and reports the result.

In this paper show that if instead of Greedy, we use Local Search in \emph{both steps}, the solution will improve significantly. Using our notation, here we are comparing LS/LS vs. GD/GD.
Figure \ref{fig:lsgd-mr-ratio} shows the improvement 
as a function of $k$, taken average over all 10 runs.

On average the improvement is $23\%$, $5.5\%$ and $6.0\%$ for GENES, MNIST10 and MNIST50 respectively. Moreover, in only 1 out of 180 runs the Greedy perfomed better than Local Search. The improvement could go as high as $67.7\%$.

\begin{figure}[!h]
\centerline{\includegraphics[width=0.45\textwidth]{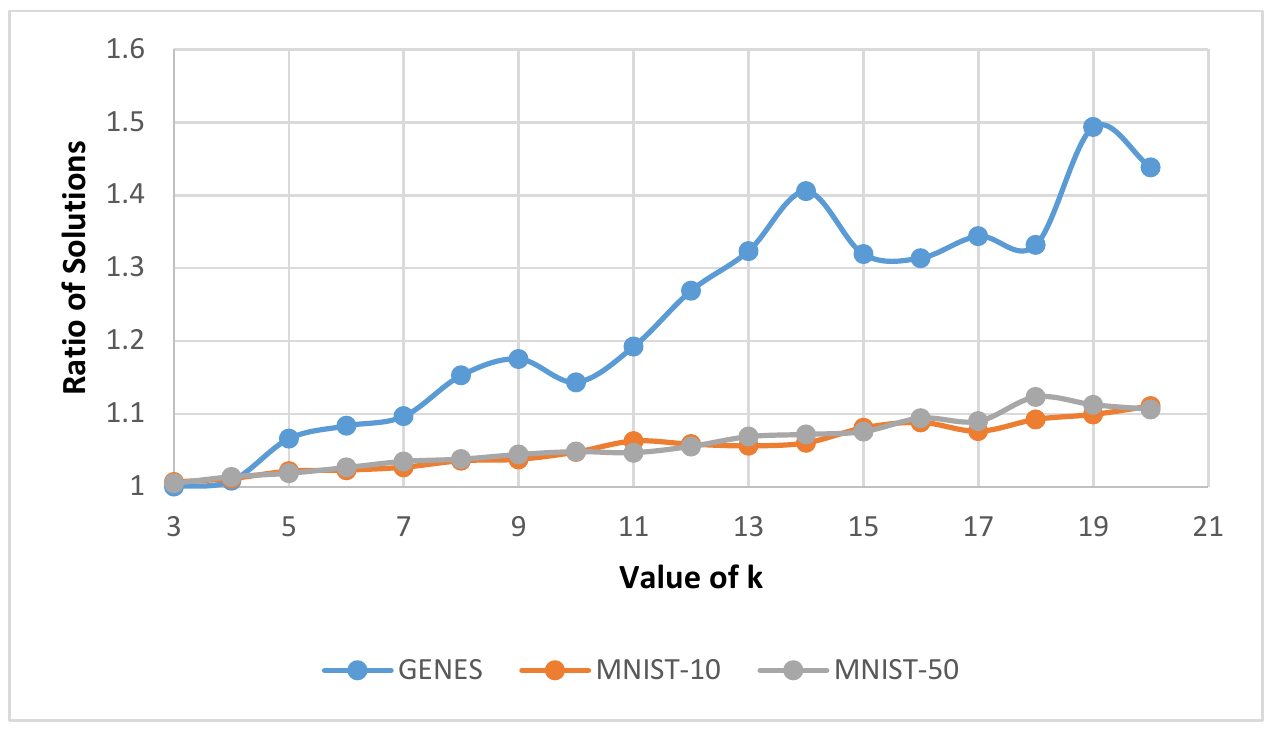}}%
\caption{Average improvement of Local Search over Greedy as a function of $k$, in the identical algorithms setting.}
\label{fig:lsgd-mr-ratio}
\end{figure}

\paragraph{Comparing Local Search vs. the LP-based algorithm.} 
In this section, we compare the performance of the Local Search algorithm and the LP-based algorithm of \cite{indyk2018composable} for constructing core-sets, i.e., we compare GD/LS with GD/LP.
Figure \ref{fig:lpls-coreset-ratio} shows how much Local Search improves over the LP-based algorithm.
On average this improvement is $7.3\%$, $1.8\%$ and $1.4\%$ for GENES, MNIST10 and MNIST50 respectively. Moreover, in $78\%$ of all runs, Local Search performed better than Lp-based algorithm, and this improvement can go upto $63\%$.
Figure \ref{fig:lpls-coreset-timeratio} shows the average ratio of the time to construct core-sets using the LP-based algorithm vs. Local Search. 
As it is clear from the graphs, our proposed Local Search algorithm performs better than even the LP-based algorithm which has almost tight approximation guarantees: while picking fewer points in the core-set, in most cases it finds a better solution and runs faster.

\begin{figure}[!h]
\minipage{0.45\textwidth}
  \includegraphics[width=\linewidth]{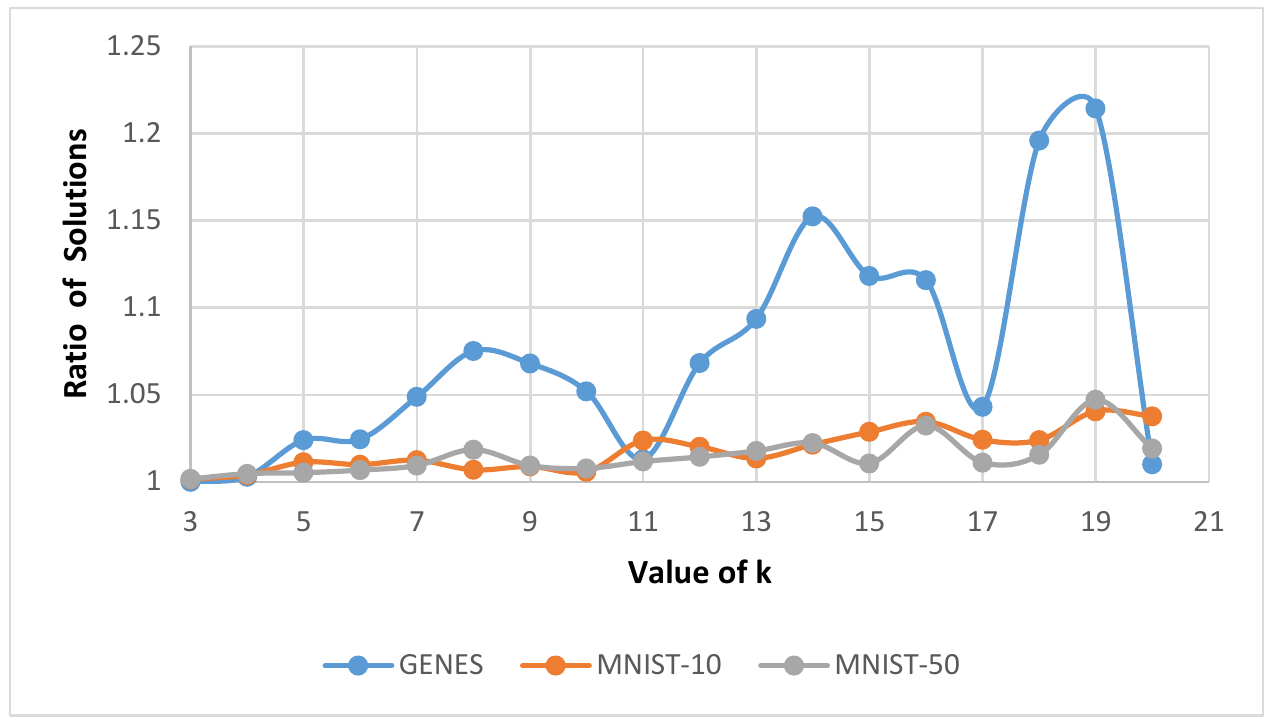}
\caption{Average improvement of Local Search over LP-based algorithm for constructing core-sets as a function of $k$.}
\label{fig:lpls-coreset-ratio}
\endminipage\hfill
\minipage{0.45\textwidth}
 \includegraphics[width=\linewidth]{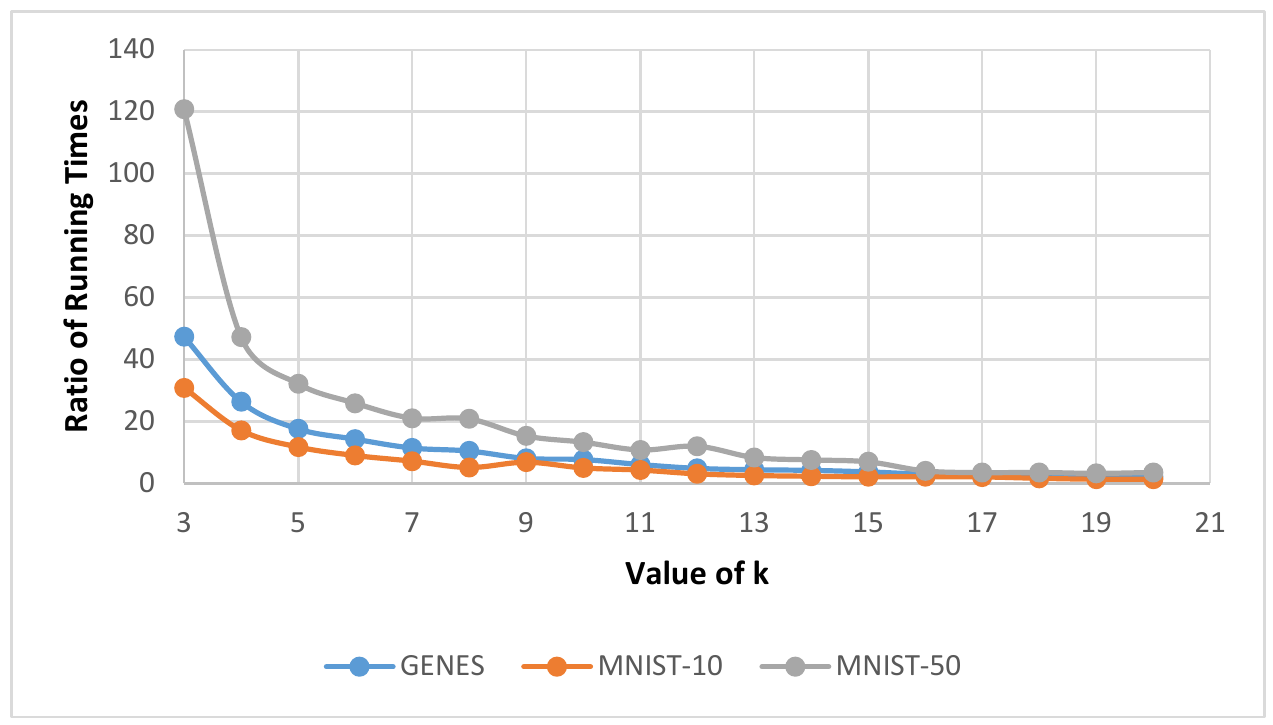}
\caption{Average ratio of the run time of the optimal algorithm over local search as a function of $k$.}
\label{fig:lpls-coreset-timeratio}
\endminipage
\end{figure}

\section{Conclusion}\label{sec:conclusion}
In this work, we proposed to use the Local Search algorithm to construct composable core-sets for the determinant maximization problem. From theoretical perspective, we showed that it achieves a near-optimal approximation guarantee. 
	We further analyzed its performance on large real data sets, and showed that most of the times, Local Search performs better than both the almost optimal approximation algorithm, and the widely-used Greedy algorithm. Generally, for larger values of $k$, the percentage of this improvement has an increasing pattern, however, the amount of this improvement depends on the data set. 
We also note that here, we used the naive implementation of the Local Search algorithm: one could tune the value of $\eps$ to further improve the quality of the solution. Finally, we provided a doubly exponential guarantee for the Greedy algorithm, however,  our experiments suggest that this bound might be open to improvement.


\section*{Acknowledgments}
The authors would like to thank Stefanie Jegelka, Chengtao Li, and Suvrit Sra for providing their data sets and source code of experiments from \cite{li2015efficient}.

Piotr Indyk was supported by NSF TRIPODS award No. 1740751 and Simons Investigator Award.
Shayan Oveis Gharan and Alireza Rezaei were supported by the NSF grant CCF-1552097 and ONR-YIP grant N00014-17-1-2429.

\bibliographystyle{alpha}
\bibliography{biblio}


\end{document}